\documentclass[12pt,a4paper]{article}
\usepackage[utf8]{inputenc}
\usepackage[T1]{fontenc}
\usepackage{amsmath,amssymb,amsthm}
\usepackage{graphicx}
\usepackage{booktabs}
\usepackage[margin=2.5cm]{geometry}
\usepackage[colorlinks=true,linkcolor=blue,citecolor=blue,urlcolor=blue]{hyperref}
\usepackage{xcolor}
\usepackage{microtype}
\theoremstyle{plain}
\newtheorem{theorem}{Theorem}[section]

\newtheorem{corollary}[theorem]{Corollary}

\theoremstyle{definition}
\newtheorem{definition}[theorem]{Definition}
\theoremstyle{remark}
\newtheorem{remark}[theorem]{Remark}

\newcommand{\R}{\mathbb{R}}
\newcommand{\N}{\mathbb{N}}
\newcommand{\Fp}{F_p}
\newcommand{\Psr}{\Psi_r}
\newcommand{\Ss}{\mathcal{S}}
\newcommand{\PT}{\mathcal{PT}}
\newcommand{\fa}{\mathfrak{a}}
\newcommand{\fb}{\mathfrak{b}}
\newcommand{\dd}{\,\mathrm{d}}

\title{Edge-cusped compactons from gradient-power nonlinear dispersion: exact transitions and singular spectral structure}

\author{Francisco R.\ Villatoro\thanks{E-mail: \href{mailto:frvillatoro@uma.es}{frvillatoro@uma.es}. ORCID: \href{https://orcid.org/0000-0003-4314-6213}{0000-0003-4314-6213}.}\\[2pt]
\small Escuela de Ingenier\'{\i}as Industriales, Universidad de M\'alaga, E-29071, M\'alaga, Spain}

\date{}

\begin{document}
\maketitle

\begin{abstract}
Degenerate nonlinear dispersion can generate compactly supported solitary waves, but the mechanisms controlling compactification, free-boundary regularity and linearised spectral structure are usually tied to the same nonlinear exponent. We introduce a conservative family depending on two nonlinear exponents, \(p>1\) and \(r>0\), and on a positive dispersion coefficient \(\mu\), with \(p\) acting on the nonlinear flux and \(r\) on its gradient. The family reduces to the absolute Rosenau--Hyman equation when \(r=1\) and \(\mu=1\), and it preserves real-valued fluxes and the compacton--anticompacton symmetry. The deformation leaves the Rosenau--Hyman amplitude--speed law unchanged, but makes the width amplitude dependent except at the distinguished value $r=1$. Two independent transition curves emerge: $pr=1$ separates infinite-tail solitary waves from compactons, whereas $r(p-1)=2$ separates compactons with a regular free boundary from edge-cusped compactons with divergent slope. The latter remain genuine distributional solutions, with no singular measures at either the free boundary or the crest. For $r\ge1$, the linearisation factors through a singular Sturm--Liouville operator whose natural Friedrichs realisation is exactly solvable: after a compacton-adapted change of variable, the eigenvalue equation is of Jacobi--hypergeometric type, with two explicit spectral sequences, exactly one negative eigenvalue and a one-dimensional translation kernel. Numerical experiments connect these exact structures to approximation and dynamics. The measured Fourier tails and projection errors agree with the rates predicted by the free-boundary exponent, while padded Fourier--Galerkin time integration propagates the cusped waves coherently at percent-level error with small phase drift. A structure-adapted Jacobi--Galerkin diagnostic of the full generator finds no detectable real part in the tested truncated point spectra. The family therefore provides an explicit testbed in which geometric singularity, weak admissibility, spectral structure and numerical regularity of compactons can be studied within a single degenerate dispersive model.
\end{abstract}

\noindent\textbf{Keywords.} Compactons; degenerate nonlinear dispersion; free boundaries; nonlinear waves; singular Sturm--Liouville operators; Rosenau--Hyman equation.

\noindent\textbf{MSC 2020.} 35Q53; 35C07; 35D30; 34B24; 35P05.

\section{Introduction}
\label{sec:intro}

Degenerate nonlinear dispersion supports wave phenomena that have no analogue in uniformly dispersive equations. Most notably, solitary waves may terminate at a finite free boundary and become compactons, so that the geometry of the support, the regularity of the profile and the degeneracy of the linearised operator are coupled. This coupling creates a useful but difficult class of nonlinear-wave problems: the same loss of regularity that produces compact support also controls weak admissibility, spectral approximation and high-frequency numerical behaviour. Explicit models in which these mechanisms can be separated are therefore valuable beyond the construction of a particular travelling wave.

The canonical example is the Rosenau--Hyman equation $K(p,p)$,
\begin{equation*}
 u_t+(u^p)_x+(u^p)_{xxx}=0,
\end{equation*}
whose compactons have an amplitude-independent width and obey a power-law amplitude--speed relation \cite{RosenauHyman1993,Rosenau1994,RosenauZilburg2018}. For non-integer $p$, however, $u^p$ is not real on negative waves. The absolute equation $|K|(p,p)$ replaces $u^p$ by the odd flux $|u|^{p-1}u$. It retains the compacton, admits real compactons and anticompactons for arbitrary real $p>1$, and has shown improved numerical robustness in collisions \cite{GarralonRusVillatoro2023CSF,GarralonRusVillatoro2023CNSNS}. These equations provide a natural setting in which to ask which features of a compacton are specific to the Rosenau--Hyman nonlinear dispersion and which survive a broader degeneracy.

We address this question with the two-exponent family
\begin{equation}
\label{eq:E}
u_t+\partial_x \Fp(u)+\mu\,\partial_x^{2}\Psr\!\bigl(\partial_x \Fp(u)\bigr)=0,
\qquad
\Fp(u)=|u|^{p-1}u,
\quad
\Psr(q)=|q|^{r-1}q,
\end{equation}
with $p>1$, $r>0$ and $\mu>0$. It reduces to $|K|(p,p)$ when $r=1$ and $\mu=1$, but the extra exponent $r$ independently tunes the degeneracy associated with the gradient of the nonlinear flux. This makes \eqref{eq:E} more than a one-parameter catalogue of compactons: it is a solvable test family in which the mechanisms governing support, edge regularity and linearised degeneracy can be varied separately. Gradient powers are also reminiscent of the $\PT$-symmetric KdV deformations introduced by Bender and collaborators \cite{BenderCooperKhareMihailaSaxena2009}; for odd integer $r$ the real branch of that construction motivates the inner gradient power in \eqref{eq:E}. We do not, however, use $\PT$ symmetry as the defining framework: for non-integer $r$ the present model is a real extension rather than a non-Hermitian deformation, and the $\PT$-symmetric compactons of \cite{BenderCooperKhareMihailaSaxena2009} belong to the Hamiltonian Cooper--Shepard--Sodano family \cite{CooperShepardSodano1993}. Broader third-order compacton classes have also been studied, including $K(f^m,g^n)$ equations \cite{IqbalNaeem2022}, while nonsmooth compact waves with singular crest regularity arise in peakompacton hierarchies \cite{ChristovKressSaxena2017}. Recent work has also systematised weak formulations for $K(m,n)$ compactons and shown that compactly supported profiles can remain meaningful beyond the classical-solution regime \cite{AncoGandarias2024}. The distinguishing feature of \eqref{eq:E} is that the second exponent acts on the gradient of the nonlinear flux itself, leading to independent control of support and edge regularity.

The exact solvability of \eqref{eq:E} exposes three distinct structural mechanisms. First, the travelling-wave problem is reduced to a single quadrature and yields an inverse-beta profile (Theorem~\ref{thm:tw}). The Rosenau--Hyman amplitude law
\begin{equation*}
 A^{p-1}=\frac{2pc}{p+1}
\end{equation*}
survives for every $r$, whereas the width is amplitude independent only on the distinguished slice $r=1$. More importantly, two different critical curves appear in the $(p,r)$ plane. The condition $pr=1$ marks the transition between infinite-tail solitary waves and compact support, while
\begin{equation*}
 r(p-1)=2
\end{equation*}
marks the transition from a regular free boundary to an edge cusp with divergent slope. Compactification and edge singularity are therefore separate phenomena.

Second, the edge cusp does not destroy the equation in the weak sense. The nonlinear dispersive flux remains sufficiently regular that the extension by zero introduces no delta or delta-prime defect at the free boundary; the loss of second-derivative regularity at the crest is likewise harmless distributionally (Theorem~\ref{thm:weak}). Thus divergent profile slope and weak admissibility can coexist in a controlled way.

Third, for $r\ge1$ the compacton-adapted linearisation has an exactly solvable self-adjoint factor. In the variable $y=(U/A)^{p-1}$ the corresponding Sturm--Liouville problem becomes a Jacobi--hypergeometric equation. The natural Friedrichs realisation has two explicit interlacing spectral sequences, exactly one negative eigenvalue and a one-dimensional kernel generated by translation, with eigenvalues independent of the speed $c$ and dispersion scale $\mu$ (Theorem~\ref{thm:spectrum}). This exact spectral information is unusual for a compacton problem with singular endpoints and a degenerate crest.

The numerical results connect these analytical structures to approximation and dynamics rather than merely checking formulas. Direct time integration propagates the $p=2$, $r=3$ and $r=5$ edge-cusped waves coherently over unit time with percent-level phase-aligned errors and small drift of the conserved quantities. More decisively, the Fourier coefficients of the exact profiles decay with the rates predicted by the free-boundary exponent, $|\widehat U_k|\sim k^{-1-\alpha}$, and the $L^2$ projection errors follow the corresponding $N^{-\alpha-1/2}$ laws. The less regular $r=5$ edge therefore produces a quantitatively slower Fourier convergence than $r=3$. Finally, a Jacobi--Galerkin truncation of the full generator $a\partial_z\Ss$ yields non-zero eigenvalues that are numerically imaginary to roundoff over the tested truncations. We use this only as evidence for modal neutrality of the truncated point spectrum, not as a stability theorem.

The broader point is that \eqref{eq:E} separates several notions that are often conflated in compacton models: compact support, free-boundary smoothness, weak solvability, the Morse index of a self-adjoint spectral factor, modal behaviour of the generator and numerical spectral regularity. In particular, the exact count $n(\Ss)=1$ does not imply stability because the equation has no evident Hamiltonian structure, and the generator remains degenerate at the free boundary and, for $r>1$, at the crest. Recent work on degenerate KdV-type equations shows that such degeneracies can produce severe high-frequency ill-posedness even when modal diagnostics are benign \cite{AmbroseSimpsonWrightYang2012,JeongOh2026,GermainHarropGriffithsMarzuola2019}. The present family supplies an explicit background on which those questions can be posed with unusually detailed travelling-wave and spectral information.

The paper is organised as follows. Section~\ref{sec:family} records the structural properties and scaling of the gradient-power family. Section~\ref{sec:tw} develops the exact travelling-wave theory and the two critical curves. Section~\ref{sec:weak} establishes weak admissibility of the cusped profiles. Section~\ref{sec:lin} derives the singular linearisation and the exact spectrum of its Friedrichs self-adjoint factor. Section~\ref{sec:num} connects the analytical singularities with static and time-dependent numerical diagnostics. Section~\ref{sec:disc} discusses the mechanisms exposed by the family and the remaining questions for the full degenerate evolution.

\section{The gradient-power family}
\label{sec:family}

Throughout, $\Fp(u)=|u|^{p-1}u$ and $\Psr(q)=|q|^{r-1}q$ with $p>1$ and $r>0$. Both functions are odd, strictly increasing and of class $C^1$ away from the origin, with $\Fp'(u)=p|u|^{p-1}$ and $\Psr'(q)=r|q|^{r-1}$. The family \eqref{eq:E} is in conservation form, so that the mass
\begin{equation}
\label{eq:mass}
M[u]=\int_{\R}u\dd x
\end{equation}
is formally conserved. Multiplication of \eqref{eq:E} by $\Fp(u)$ shows that
\begin{equation}
\label{eq:Q}
Q[u]=\frac{1}{p+1}\int_{\R}|u|^{p+1}\dd x
\end{equation}
is also formally conserved. Indeed $\Fp(u)u_t=\partial_t|u|^{p+1}/(p+1)$, $\Fp(u)\partial_x\Fp(u)=\partial_x\Fp(u)^2/2$ and, writing $F=\Fp(u)$,
\begin{equation}
\label{eq:Qident}
F\,\partial_x^2\Psr(F_x)=\partial_x\bigl(F\,\partial_x\Psr(F_x)\bigr)-\frac{r}{r+1}\,\partial_x|F_x|^{r+1},
\end{equation}
because $F_x\partial_x\Psr(F_x)=r|F_x|^{r-1}F_xF_{xx}=\tfrac{r}{r+1}\partial_x|F_x|^{r+1}$. These are the two invariants of $|K|(p,p)$ that persist for all $r$; for $r=1$ the equation has additional invariants \cite{GarralonRusVillatoro2023CSF} whose fate for $r\neq1$ we do not address.

The family is invariant under $u\mapsto-u$, so that any solution has a mirror image of opposite sign; in particular each compacton has an anticompacton. It is also invariant under the combined transformation $x\mapsto -x$, $t\mapsto -t$, which is the $\PT$ symmetry in the convention of \cite{BenderCooperKhareMihailaSaxena2009} with $u$ unchanged; each of the three terms in \eqref{eq:E} changes sign under it. Finally, \eqref{eq:E} has the two-parameter scaling group
\begin{equation}
\label{eq:scaling}
u\mapsto \lambda u,\qquad x\mapsto \lambda^{p(r-1)/(r+1)}x,\qquad t\mapsto \lambda^{p(r-1)/(r+1)+1-p}t,
\end{equation}
which is obtained by balancing the three terms. The exponent of $x$ vanishes precisely for $r=1$, which is the origin of the amplitude-independent width of the Rosenau--Hyman compacton; for $r>1$ wider compactons are taller.

For $r>1$ the highest-order term of \eqref{eq:E} is degenerate wherever $u=0$ and also wherever $\partial_x\Fp(u)=0$, since its coefficient in quasilinear form is proportional to $|u|^{p-1}|\partial_x\Fp(u)|^{r-1}$. The second kind of degeneracy is absent for $r=1$ and is the main new feature of the family; in the compacton it occurs at the crest. We do not know of a Hamiltonian formulation of \eqref{eq:E} for $r\neq1$, and the Lagrangian structure that underlies the $\PT$-symmetric compactons of \cite{BenderCooperKhareMihailaSaxena2009} does not apply to the Rosenau--Hyman form of the nonlinear dispersion.

\section{Exact travelling waves}
\label{sec:tw}

We look for solutions $u(x,t)=U(z)$, $z=x-ct$, with $c>0$, in the class of even, non-negative profiles that are non-increasing on $z>0$, with $U$ and its fluxes vanishing in the vacuum. On the support we write $W=U^p=\Fp(U)$. Since $U\ge0$, \eqref{eq:E} becomes
\begin{equation}
\label{eq:twode}
-cU'+W'+\mu\bigl[\Psr(W')\bigr]''=0,
\end{equation}
where the prime denotes $\dd/\dd z$. We use the standard notation $B(\fa,\fb)$ for the Euler beta function, $B_y(\fa,\fb)=\int_0^y s^{\fa-1}(1-s)^{\fb-1}\dd s$ for the incomplete beta function and $I_y(\fa,\fb)=B_y(\fa,\fb)/B(\fa,\fb)$ for its regularised form, and we introduce the exponents
\begin{equation}
\label{eq:abparams}
\fa=\frac{pr-1}{(p-1)(r+1)},\qquad \fb=\frac{r}{r+1},\qquad \alpha=\frac{r+1}{pr-1}.
\end{equation}

\begin{theorem}[Exact compactons]
\label{thm:tw}
Let $p>1$, $r>0$, $\mu>0$ and $c>0$, and assume $pr>1$. Up to translation, \eqref{eq:E} has exactly one travelling wave $u=U(x-ct)$ in the class described above with compact support, and the following hold.
\begin{enumerate}
\item[(i)] The amplitude is
\begin{equation}
\label{eq:amplitude}
A=U(0)=\Bigl(\frac{2pc}{p+1}\Bigr)^{1/(p-1)},
\end{equation}
independently of $r$ and $\mu$, so that the Rosenau--Hyman amplitude--speed law is preserved.
\item[(ii)] The support is $[-\ell,\ell]$ with half-width
\begin{equation}
\label{eq:width}
\ell=\Bigl(\frac{2\mu r}{r+1}\Bigr)^{1/(r+1)}\frac{p}{p-1}\,A^{p(r-1)/(r+1)}\,B(\fa,\fb).
\end{equation}
\item[(iii)] The profile is given, for $|z|\le\ell$, by
\begin{equation}
\label{eq:profile}
U(z)=A\,y(z)^{1/(p-1)},\qquad I_{y(z)}(\fa,\fb)=1-\frac{|z|}{\ell},
\end{equation}
and $U(z)=0$ for $|z|>\ell$. The profile is strictly decreasing on $(0,\ell)$.
\item[(iv)] At the free boundary, with $d=\ell-|z|\downarrow0$,
\begin{equation}
\label{eq:edge}
U=C_e\,d^{\alpha}\bigl(1+o(1)\bigr),\qquad C_e=A\Bigl(\frac{\fa\,B(\fa,\fb)}{\ell}\Bigr)^{\alpha}.
\end{equation}
Consequently $U'(\pm\ell)=0$ if $r(p-1)<2$, $U'$ has a finite nonzero one-sided limit if $r(p-1)=2$, and $|U'(z)|\to\infty$ as $|z|\to\ell$ if $r(p-1)>2$.
\item[(v)] At the crest, $U=A-\kappa|z|^{1+1/r}(1+o(1))$ with $\kappa=\frac{A}{p-1}\bigl(\fb B(\fa,\fb)/\ell\bigr)^{1+1/r}$. Thus $U\in C^{1,\min\{1,1/r\}}$ near $z=0$; in particular $U''$ is unbounded at the crest when $r>1$, the crest is quadratic when $r=1$, and the leading correction has vanishing second derivative at the crest when $0<r<1$.
\end{enumerate}
\end{theorem}

The three quantities $\fa$, $\fb$ and $\alpha$ satisfy $\alpha\,\fa\,(p-1)=1$, and the two curves $pr=1$ and $r(p-1)=2$ in the $(p,r)$ plane are the loci $\fa=0$ and $\alpha=1$; see Figure~\ref{fig:profiles} for the case $p=2$.

\begin{proof}
Integrating \eqref{eq:twode} once and using the vanishing of $U$, $W$ and $\Psr(W')$ in the vacuum gives
\begin{equation}
\label{eq:int1}
-cU+W+\mu\bigl[\Psr(W')\bigr]'=0.
\end{equation}
Multiplying by $W'$, using $UW'=pU^pU'=\tfrac{p}{p+1}(W^{(p+1)/p})'$ and $[\Psr(W')]'W'=\tfrac{r}{r+1}(|W'|^{r+1})'$, and integrating once more with the same vanishing conditions, we obtain the first integral
\begin{equation}
\label{eq:firstint}
\frac{\mu r}{r+1}|W'|^{r+1}=\frac{cp}{p+1}W^{(p+1)/p}-\frac12W^{2}.
\end{equation}
The right-hand side is positive for $0<W<W_*$ and vanishes at $W_*$, where $W_*^{(p-1)/p}=2pc/(p+1)$. Since $W'=0$ at the crest, $W(0)=W_*$ and $A=W_*^{1/p}$, which is \eqref{eq:amplitude}. The value $W\equiv W_*$ cannot be attained on an interval, because \eqref{eq:int1} at $W=W_*$ gives $[\Psr(W')]'=(cA-A^p)/\mu\neq0$ whereas $\Psr(W')\equiv0$ there; hence flat tops are excluded and $W'<0$ on $(0,\ell)$, which yields strict monotonicity in (iii).

Introduce $m=(p-1)/p$ and $y=(W/W_*)^{m}=(U/A)^{p-1}$. Using $cp/(p+1)=W_*^{m}/2$, \eqref{eq:firstint} becomes $|W'|^{r+1}=\tfrac{r+1}{2\mu r}W^{(p+1)/p}(W_*^m-W^m)$. On $z>0$, where $W'<0$, the distance $s=\ell-z$ from the boundary to a point where $W$ takes a given value is
\begin{equation}
\label{eq:quad}
s=\int_0^{W}\frac{\dd\omega}{|W'(\omega)|}
=\Bigl(\frac{2\mu r}{r+1}\Bigr)^{\frac{1}{r+1}}\int_0^{W}\omega^{-\frac{p+1}{p(r+1)}}\bigl(W_*^m-\omega^m\bigr)^{-\frac{1}{r+1}}\dd\omega .
\end{equation}
The substitution $\omega=W_*y^{1/m}$ gives $\dd\omega=(W_*/m)y^{1/m-1}\dd y$, and a direct computation of the exponents shows that the integrand becomes $y^{\fa-1}(1-y)^{\fb-1}$ up to the constant factor
\begin{equation}
\label{eq:Cs}
C_s=\Bigl(\frac{2\mu r}{r+1}\Bigr)^{\frac{1}{r+1}}\frac{p}{p-1}W_*^{\,m-\frac{p+1}{p(r+1)}-\frac{m}{r+1}}
=\Bigl(\frac{2\mu r}{r+1}\Bigr)^{\frac{1}{r+1}}\frac{p}{p-1}A^{\,p\frac{r-1}{r+1}}.
\end{equation}
Hence $s=C_s B_y(\fa,\fb)$. The integral converges at $y=0$ if and only if $\fa>0$, that is $pr>1$, which is the compactness condition; evaluating at $y=1$ gives $\ell=C_sB(\fa,\fb)$, which is \eqref{eq:width}, and dividing gives \eqref{eq:profile}. Uniqueness follows because \eqref{eq:int1} and \eqref{eq:firstint} are forced by the vanishing conditions and, on $(0,\ell)$, the sign of $W'$ and the quadrature determine $W$; evenness then determines the other half.

For (iv), $I_y(\fa,\fb)=y^{\fa}/(\fa B(\fa,\fb))\,(1+O(y))$ as $y\to0$, so that $d/\ell=y^{\fa}/(\fa B)\,(1+O(y))$, whence $y=(\fa B d/\ell)^{1/\fa}(1+o(1))$ and $U=Ay^{1/(p-1)}$ gives \eqref{eq:edge} with $\alpha=1/(\fa(p-1))=(r+1)/(pr-1)$. The three regimes follow from $\alpha>1$, $\alpha=1$ and $\alpha<1$, and $\alpha=1$ is equivalent to $r(p-1)=2$. For (v), $1-I_y(\fa,\fb)=(1-y)^{\fb}/(\fb B)\,(1+O(1-y))$ as $y\to1$, so $|z|/\ell=(1-y)^{\fb}/(\fb B)\,(1+o(1))$, $1-y=(\fb B|z|/\ell)^{1/\fb}(1+o(1))$, and $1/\fb=1+1/r$; expanding $U=A y^{1/(p-1)}=A(1-(1-y)/(p-1)+\dots)$ gives the stated expansion.
\end{proof}

\begin{remark}[Recovery of the Rosenau--Hyman compacton]
For $r=1$ and $\mu=1$ one has $\fa=\fb=1/2$, $B(\tfrac12,\tfrac12)=\pi$ and $I_y(\tfrac12,\tfrac12)=\tfrac{2}{\pi}\arcsin\sqrt{y}$, so that $y=\cos^2\bigl(\tfrac{p-1}{2p}z\bigr)$ and \eqref{eq:profile} becomes $U=A\cos^{2/(p-1)}\bigl(\tfrac{p-1}{2p}z\bigr)$ with $\ell=\pi p/(p-1)$, the classical compacton of $K(p,p)$ and $|K|(p,p)$ \cite{RosenauHyman1993,GarralonRusVillatoro2023CSF}. For $p=2$ this is $U=\tfrac{4c}{3}\cos^2(z/4)$ with $\ell=2\pi$.
\end{remark}

\begin{remark}[Solitary waves without compact support]
If $pr\le1$ the quadrature \eqref{eq:quad} diverges at $W=0$ and the profile has an infinite tail. From \eqref{eq:firstint}, $|W'|\sim CW^{(p+1)/(p(r+1))}$ as $W\to0$, so the tail decays algebraically, $W\sim z^{-p(r+1)/(1-pr)}$, for $pr<1$, and exponentially for $pr=1$. We do not pursue this regime further.
\end{remark}

\begin{remark}[The case $p=2$]
For $p=2$ the amplitude is $A=4c/3$ for every $r$, the width is $\ell\propto c^{(r-1)/(r+1)}$ and the edge exponent is $\alpha=(r+1)/(2r-1)$. The compacton is edge-cusped for $r>2$; for $r=3$ one has $\alpha=4/5$ and $U'\sim d^{-1/5}$, and for $r=5$ one has $\alpha=2/3$ and $U'\sim d^{-1/3}$. With $c=\mu=1$ the total widths are $L=4\pi\approx12.566$, $L\approx6.5558$ and $L\approx5.3703$ for $r=1,3,5$ respectively, so that the deformation narrows the compacton at fixed speed.
\end{remark}

\begin{figure}[t]
\centering
\includegraphics[width=\textwidth]{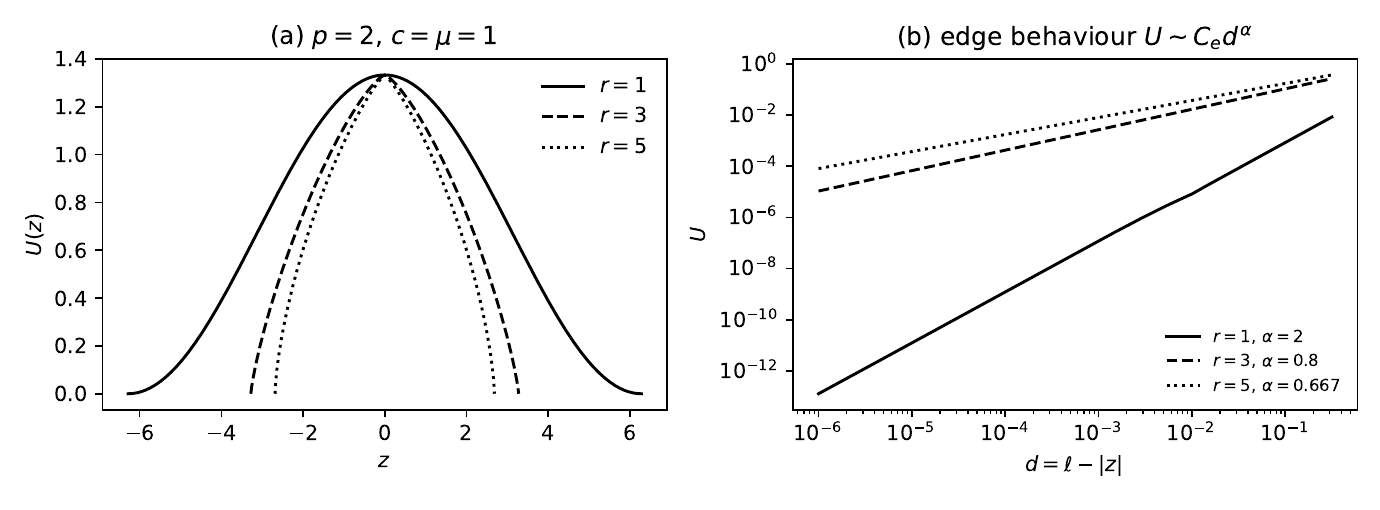}
\caption{(a) Compacton profiles of \eqref{eq:E} for $p=2$, $c=\mu=1$ and $r=1,3,5$, computed from \eqref{eq:profile}. (b) Behaviour near the free boundary in logarithmic scale; the slopes are the exponents $\alpha=2$, $4/5$ and $2/3$ of \eqref{eq:edge}. For $r=3$ and $r=5$ the slope is smaller than one and the derivative of the profile is unbounded at the boundary.}
\label{fig:profiles}
\end{figure}

\section{Weak admissibility of the edge-cusped compactons}
\label{sec:weak}

For $r(p-1)>2$ the compacton of Theorem~\ref{thm:tw} has an unbounded derivative at the free boundary, and for $r>1$ it has an unbounded second derivative at the crest, so it is not a classical solution of \eqref{eq:E}. We show that it is nevertheless a solution in the sense of distributions and, more precisely, that every term of the equation is a locally integrable function.

\begin{definition}
\label{def:weak}
Let $T>0$. A function $u\in L^\infty(\R\times(0,T))$ is a weak solution of \eqref{eq:E} if $\partial_x\Fp(u)\in L^1_{\rm loc}(\R\times(0,T))$, $\Psr(\partial_x\Fp(u))\in L^1_{\rm loc}(\R\times(0,T))$ and, for every $\varphi\in C_c^\infty(\R\times(0,T))$,
\begin{equation}
\label{eq:weakform}
\int_0^T\!\!\int_{\R}\Bigl[u\,\varphi_t+\Fp(u)\,\varphi_x-\mu\,\Psr\bigl(\partial_x\Fp(u)\bigr)\varphi_{xx}\Bigr]\dd x\dd t=0 .
\end{equation}
\end{definition}

\begin{theorem}[Weak admissibility]
\label{thm:weak}
Let $U$ be the compacton of Theorem~\ref{thm:tw} extended by zero outside $[-\ell,\ell]$, and let $u(x,t)=U(x-ct)$. Then $G=\Psr(\partial_x\Fp(u))$ belongs to $C^1(\R)$ in $x$ for each $t$, its second derivative $G_{xx}$ is in $L^1(\R)$, and $u_t$, $\partial_x\Fp(u)$ and $\mu G_{xx}$ are functions in $L^1(\R)$ that satisfy \eqref{eq:E} almost everywhere. In particular $u$ is a weak solution in the sense of Definition~\ref{def:weak}, and no singular measure is supported on the free boundary or on the crest.
\end{theorem}

\begin{proof}
Work in the variable $z$. On $(-\ell,\ell)$ set $W=U^p$ and $G=\Psr(W')$. By Theorem~\ref{thm:tw}(iv), $W=U^p\sim C_e^pd^{p\alpha}$ near the boundary, and $p\alpha=p(r+1)/(pr-1)>1$ because $pr-1<p(r+1)$; hence $W'\to0$ and $G=\Psr(W')\to0$ at $z=\pm\ell$. By \eqref{eq:int1}, $G'=(cU-W)/\mu$ on $(-\ell,\ell)$, which is continuous on $[-\ell,\ell]$ and vanishes at $\pm\ell$. Therefore the zero extension of $G$ is continuous with a continuous first derivative on $\R$, that is $G\in C^1(\R)$, and no jump of $G$ or of $G'$ occurs at the boundary; in particular no $\delta$ or $\delta'$ contribution arises when $G''$ is computed in the sense of distributions. At the crest, $U$ and $W$ are continuous, $U'(0)=W'(0)=0$ by Theorem~\ref{thm:tw}(v), and therefore $G'=(cU-W)/\mu$ is continuous across $z=0$ with $G''=(cU'-W')/\mu$ continuous there as well; the unboundedness of $U''$ at the crest is invisible in $G$. On $(-\ell,\ell)\setminus\{0\}$ the distributional derivative $G''$ is the function $(cU'-W')/\mu$, which is in $L^1$ because $U'\sim d^{\alpha-1}$ with $\alpha>0$ is integrable and $W'$ is bounded. Collecting, $-cU'$, $W'$ and $\mu G''$ are $L^1$ functions on $\R$ whose sum vanishes identically on $(-\ell,\ell)$ by \eqref{eq:twode} and outside the support because all three vanish there. Multiplying by a test function and integrating by parts, which is legitimate because $G\in C^1$ with $G''\in L^1$, gives \eqref{eq:weakform}.
\end{proof}

\begin{corollary}
\label{cor:reg}
The compacton satisfies $U\in C^{0,\min\{\alpha,1\}}(\R)$, and $U\in C^{1,\min\{\alpha-1,1/r,1\}}(\R)$ when $r(p-1)<2$. Its derivative belongs to $L^{q}(\R)$ for $q<1/(1-\alpha)$ when $\alpha<1$, so that $U\in W^{1,1}(\R)$ but $U\notin W^{1,\infty}(\R)$ in the edge-cusped regime.
\end{corollary}

Theorem~\ref{thm:weak} is a statement about the travelling wave only. It says nothing about uniqueness of weak solutions with compacton initial data, a question that is delicate already for $K(p,p)$ and that we do not address.

\section{Linearisation and exact spectrum}
\label{sec:lin}

\subsection{The linearised generator}

The travelling-wave and weak-solution results above hold for every $r>0$.  In this section we restrict to
\begin{equation}
\label{eq:spectralrange}
r\ge 1.
\end{equation}
This restriction is intrinsic to the conventional linearisation used below.  Indeed, the dispersive constitutive law $\Psi_r(q)=|q|^{r-1}q$ is differentiable at $q=0$ only for $r\ge1$: for $0<r<1$ the quotient $\Psi_r(\varepsilon h)/\varepsilon=\varepsilon^{r-1}|h|^{r-1}h$ diverges as $\varepsilon\downarrow0$.  Since $(U^p)'=0$ at the compacton crest, the spectral problem below is therefore not the linearisation of \eqref{eq:E} when $0<r<1$.  The nonlinear travelling waves in that range remain covered by Theorems~\ref{thm:tw} and~\ref{thm:weak}.

In the frame moving with the compacton, \eqref{eq:E} reads $u_t-cu_z+\partial_z\Fp(u)+\mu\partial_z^2\Psr(\partial_z\Fp(u))=0$. Let $u=U+v$ with $v$ small and supported in $(-\ell,\ell)$, where $U>0$. Setting $a(z)=\Fp'(U)=pU^{p-1}$, one has $\Fp(u)=W+av+O(v^2)$, and it is convenient to use the perturbation of the flux, $\eta=av$, as dependent variable. With $b(z)=\Psr'(W')=r|W'|^{r-1}$, the linearised equation is
\begin{equation}
\label{eq:linraw}
\partial_t\Bigl(\frac{\eta}{a}\Bigr)-c\,\partial_z\Bigl(\frac{\eta}{a}\Bigr)+\eta_z+\mu\,\partial_z^2\bigl(b\,\eta_z\bigr)=0 .
\end{equation}
Since $a$ does not depend on $t$, this can be written as
\begin{equation}
\label{eq:gen}
\eta_t=a\,\partial_z\Ss\eta,\qquad \Ss\eta=-\mu\bigl(b\,\eta_z\bigr)_z+\Bigl(\frac{c}{a}-1\Bigr)\eta .
\end{equation}
The operator $\Ss$ is formally symmetric in $L^2(-\ell,\ell)$ with respect to $\dd z$, while $a\partial_z$ is formally skew-symmetric in $L^2$ with respect to the weight $a^{-1}\dd z$; the two structures do not match, which reflects the absence of a Hamiltonian formulation. Differentiating \eqref{eq:int1} with respect to $z$ and using $U'=W'/a$ gives $\Ss W'=0$, so the translation mode belongs to the kernel of $\Ss$, as it must.

The coefficient $b$ is positive on $(-\ell,\ell)\setminus\{0\}$, vanishes at the crest like $|z|^{(r-1)/r}$ when $r>1$ by Theorem~\ref{thm:tw}(v), and vanishes at the free boundary like $d^{(p\alpha-1)(r-1)}$ when $r>1$; the potential $c/a-1$ tends to $+\infty$ at the free boundary like $d^{-(p-1)\alpha}$ and is bounded below on the support. Thus $\Ss$ is a singular Sturm--Liouville operator with a degenerate interior point and two singular endpoints.

\subsection{Reduction to a hypergeometric equation}

The key observation is that the change of independent variable $y=(U/A)^{p-1}$ of the proof of Theorem~\ref{thm:tw} also linearises the coefficients of $\Ss$. Indeed $a=pA^{p-1}y$, so that
\begin{equation}
\label{eq:cova}
\frac{c}{a}=\frac{p+1}{2p^2}\,\frac{1}{y},
\end{equation}
and a direct computation from the first integral \eqref{eq:firstint} gives
\begin{equation}
\label{eq:by}
\mu\,b\,(y_z)^2=D_{p,r}\,y(1-y),\qquad D_{p,r}=\frac{(r+1)(p-1)^2}{2p^2},
\end{equation}
together with $\mu\,y_z\,\partial_y(b\,y_z)=D_{p,r}\bigl[\fa-(\fa+\fb)y\bigr]$. Writing $\eta_z=y_z\eta_y$ and $(b\eta_z)_z=y_z\partial_y(by_z\eta_y)$, the eigenvalue equation $\Ss\eta=\lambda\eta$ becomes, on each half of the support,
\begin{equation}
\label{eq:hyp}
y(1-y)\eta_{yy}+\bigl[\fa-(\fa+\fb)y\bigr]\eta_y+\Bigl[\frac{\lambda+1}{D_{p,r}}-\frac{p+1}{(r+1)(p-1)^2}\,\frac1y\Bigr]\eta=0 .
\end{equation}
This is a Gauss hypergeometric equation with an additional $1/y$ term. The indicial equation at $y=0$ has the roots
\begin{equation}
\label{eq:indicial}
s_+=\gamma=\frac{p+1}{(p-1)(r+1)},\qquad s_-=-\frac{1}{p-1},
\end{equation}
and the substitution $\eta=y^{\gamma}f(y)$ transforms \eqref{eq:hyp} into $y(1-y)f''+[C-(A+B+1)y]f'-ABf=0$ with
\begin{equation}
\label{eq:hypparams}
C=\frac{pr+2p+1}{(p-1)(r+1)},\qquad A+B=S:=\frac{pr+p+2}{(p-1)(r+1)},\qquad AB=\gamma(S-\gamma)-\frac{\lambda+1}{D_{p,r}} .
\end{equation}
Note that $C=2\gamma+\fa$ and $C-A-B=1/(r+1)$, so the Frobenius exponents at $y=1$, the crest, are $0$ and $1/(r+1)$. The measure $\dd z$ corresponds to $\dd z=C_s\,y^{\fa-1}(1-y)^{\fb-1}\dd y$ by \eqref{eq:quad}, so that $L^2(-\ell,\ell)$ is isometric, up to the constant $C_s$, to the direct sum of two copies of $L^2\bigl((0,1),y^{\fa-1}(1-y)^{\fb-1}\dd y\bigr)$, one for each half of the support.

\subsection{The self-adjoint realisation}

Let
\begin{equation}
\label{eq:form}
q[\eta]=\int_{-\ell}^{\ell}\Bigl[\mu\,b\,|\eta_z|^2+\Bigl(\frac{c}{a}-1\Bigr)|\eta|^2\Bigr]\dd z,\qquad \eta\in C_c^\infty(-\ell,\ell).
\end{equation}
For $r\ge1$, the potential $V=c/a-1$ is bounded below and $b\ge0$.  Hence, after adding a sufficiently large multiple of $\|\eta\|_{L^2}^2$, the standard singular Sturm--Liouville form construction \cite{Zettl2005} gives a non-negative closable form.  We denote its closure by $\overline q$ and by $\Ss_F$ the self-adjoint operator associated with $\overline q$; equivalently, $\Ss_F$ is the Friedrichs extension of the corresponding minimal symmetric Sturm--Liouville operator.  Below we write simply $\Ss=\Ss_F$.

The form domain determines the transmission condition at the degenerate crest.  By Theorem~\ref{thm:tw}(v),
\[
 b(z)\asymp |z|^\kappa,\qquad \kappa=\frac{r-1}{r}\in[0,1)
\]
near $z=0$, so that $b^{-1}$ is locally integrable.  For smooth $\eta$ and $x,y$ close to the crest,
\begin{equation}
\label{eq:crestholder}
|\eta(x)-\eta(y)|^2
\le
\left(\int_x^y b|\eta_z|^2\dd z\right)
\left(\int_x^y b^{-1}\dd z\right).
\end{equation}
The second factor tends to zero with $|x-y|$, uniformly for sets bounded in the form norm.  By closure, every element of $D(\overline q)$ therefore has a continuous representative across the crest.  For $\eta$ in the operator domain, Green's identity applied with test functions crossing $z=0$ then gives the natural flux transmission condition
\begin{equation}
\label{eq:transmission}
[\eta]_{0^-}^{0^+}=0,\qquad [b\eta_z]_{0^-}^{0^+}=0.
\end{equation}

The Friedrichs form also selects the principal Frobenius branch at each free boundary.  In the variable $y$, the derivative part of the form is, up to a positive constant,
\[
 \int_0^1 y^{\fa}(1-y)^{\fb}\,|\eta_y|^2\dd y,
\]
and the singular part of the potential has the same integrability threshold.  A branch $\eta\sim y^s$ therefore has finite Friedrichs energy only if
\begin{equation}
\label{eq:formbranch}
\fa+2s>1.
\end{equation}
For $s=s_+=\gamma$ this holds because $\fa+2\gamma=C>1$, whereas
\begin{equation}
\label{eq:singularenergy}
\fa+2s_--1=-\frac{p+r+2}{(p-1)(r+1)}<0.
\end{equation}
Thus the singular branch is never in the Friedrichs form domain, even in parameter regimes in which it is square integrable.

Finally, the resolvent of $\Ss$ is compact.  Let $(\eta_j)$ be bounded in the shifted form norm.  On compact subintervals away from the free boundaries, \eqref{eq:crestholder} and the local integrability of $b^{-1}$ give a uniform modulus of continuity, including across the crest; together with the $L^2$ bound this yields local uniform boundedness, and Arzel\`a--Ascoli gives compactness in $L^2$ on every such subinterval.  Near either free boundary, $V(z)\to+\infty$, and therefore
\[
 \int_{\{d(z)<\delta\}}|\eta_j|^2\dd z
 \le
 \frac{1}{\inf_{\{d<\delta\}} V}
 \int_{\{d<\delta\}}V_+|\eta_j|^2\dd z
 \longrightarrow0
 \qquad(\delta\downarrow0)
\]
uniformly in $j$, since the positive part of the potential is controlled by the shifted form norm.  Local compactness plus this uniform endpoint tightness proves that $D(\overline q)$ embeds compactly in $L^2(-\ell,\ell)$ and hence that $\Ss$ has compact resolvent.

\begin{remark}[Limit-point and limit-circle endpoints]
\label{rem:lplc}
The branch $\eta\sim y^{s_-}$ belongs to $L^2(\dd z)$ near a free boundary if and only if $2s_-+\fa>0$, namely
\begin{equation}
\label{eq:lplcthreshold}
pr>2r+3.
\end{equation}
Thus for $p=2$ the free boundary is always limit-point.  When $p>2+3/r$ it is limit-circle, both Frobenius branches are square integrable and a self-adjoint realisation requires an endpoint condition.  The Friedrichs choice remains unambiguous: by \eqref{eq:singularenergy}, the singular branch has infinite form energy for every $p>1$ and $r\ge1$, even when it lies in $L^2$.  Hence the Friedrichs realisation always selects the principal branch $y^{s_+}$.  Theorem~\ref{thm:spectrum} refers to this realisation.
\end{remark}

\begin{theorem}[Exact spectrum of $\Ss$]
\label{thm:spectrum}
Let $p>1$, $r\ge1$, and let $\Ss$ be the Friedrichs realisation associated with \eqref{eq:form}. Then $\Ss$ has compact resolvent, and its spectrum consists of the two sequences
\begin{align}
\label{eq:even}
\lambda_n^{(+)}&=-\frac{(p-1)(pr+2p+1)}{2p^2(r+1)}+\frac{(p-1)\bigl[p(r+1)+2\bigr]}{2p^2}\,n+\frac{(r+1)(p-1)^2}{2p^2}\,n^2,\\
\label{eq:odd}
\lambda_n^{(-)}&=\frac{(p-1)(r+3)}{2p}\,n+\frac{(r+1)(p-1)^2}{2p^2}\,n^2,\qquad n=0,1,2,\dots,
\end{align}
with even and odd eigenfunctions respectively,
\begin{align}
\label{eq:evenfun}
\eta_n^{(+)}(z)&=y^{\gamma}\,{}_2F_1\bigl(-n,\,S+n;\,C;\,y\bigr),\\
\label{eq:oddfun}
\eta_n^{(-)}(z)&=\operatorname{sgn}(z)\,y^{\gamma}(1-y)^{1/(r+1)}\,{}_2F_1\bigl(-n,\,C+n+\tfrac{1}{r+1};\,C;\,y\bigr),
\end{align}
where $y=y(z)$ is the profile variable of \eqref{eq:profile}. All eigenvalues are simple, they do not depend on $c$ or $\mu$, and
\begin{equation}
\label{eq:index}
\lambda_0^{(+)}<0=\lambda_0^{(-)}<\lambda_1^{(+)}<\lambda_1^{(-)}<\cdots,\qquad \ker\Ss=\operatorname{span}\{W'\}.
\end{equation}
In particular $\Ss$ has exactly one negative eigenvalue.
\end{theorem}

\begin{proof}
The self-adjointness, transmission conditions, endpoint selection and compactness of the resolvent were established above.  Hence $\Ss$ has a complete orthonormal set of eigenfunctions.  Since the coefficients are even, the even and odd subspaces reduce $\Ss$, and every eigenfunction can be chosen to have definite parity. Any eigenfunction is, on each half of the support, a solution of \eqref{eq:hyp} in the form domain, hence of the form $y^{\gamma}f$ with $f$ a solution of the hypergeometric equation with parameters \eqref{eq:hypparams} that is regular at $y=0$, by the Friedrichs selection of the principal branch. Near the crest, $f$ is a combination of the two Frobenius solutions, $f_1$ analytic at $y=1$ with $f_1(1)=1$ and $f_2=(1-y)^{1/(r+1)}g$ with $g$ analytic at $y=1$ and $g(1)=1$. Using \eqref{eq:by}, one computes that $b\,\partial_z\bigl(y^\gamma f_1\bigr)\to0$ as $z\to0$, whereas $b\,\partial_z\bigl(y^\gamma f_2\bigr)$ tends to a finite nonzero limit, of opposite signs on the two sides. Even functions have odd flux $b\eta_z$, so the transmission condition $[b\eta_z]=0$ forces the coefficient of $f_2$ to vanish; odd functions vanish at the crest, so the coefficient of $f_1$ must vanish. Hence even eigenfunctions correspond to solutions of the hypergeometric equation that are regular at both $y=0$ and $y=1$, which requires $A=-n$ or $B=-n$ for some $n\in\N_0$, and odd eigenfunctions correspond to the branch $(1-y)^{C-A-B}$ at $y=1$ with regular remaining factor, which requires $C-A=-n$ or $C-B=-n$. Solving \eqref{eq:hypparams} for $\lambda$ under $A=-n$, $B=S+n$ gives \eqref{eq:even}, and under $A=C+n$, $B=S-C-n$ gives \eqref{eq:odd}; the corresponding eigenfunctions are \eqref{eq:evenfun} and \eqref{eq:oddfun}. Because the resolvent is compact, self-adjoint spectral theory supplies a complete eigenbasis.  The preceding parity argument classifies every possible eigenfunction: in the even sector the transmission condition eliminates the second crest branch and the connection coefficient vanishes only when $A=-n$ or $B=-n$; in the odd sector continuity eliminates the first crest branch and the corresponding connection coefficient vanishes only when $C-A=-n$ or $C-B=-n$.  Consequently \eqref{eq:even}--\eqref{eq:oddfun} exhaust the spectrum.  Equivalently, after the factor $y^\gamma$ (and, in the odd sector, $(1-y)^{1/(r+1)}$) is removed, the eigenfunctions reduce to Jacobi polynomials with admissible endpoint exponents.

The independence of $c$ and $\mu$ is manifest in \eqref{eq:even} and \eqref{eq:odd}.  The strict interlacing in \eqref{eq:index} follows directly from
\begin{align}
\label{eq:interlace1}
\lambda_n^{(-)}-\lambda_n^{(+)}
&=\frac{p-1}{2p^2(r+1)}
\Bigl[2n(p-1)(r+1)+pr+2p+1\Bigr]>0,\\
\label{eq:interlace2}
\lambda_{n+1}^{(+)}-\lambda_n^{(-)}
&=\frac{r(p-1)}{2p^2(r+1)}
\Bigl[2n(p-1)(r+1)+2pr+3p-r\Bigr]>0,
\end{align}
for every $n\ge0$.  In particular $\lambda_0^{(+)}<\lambda_0^{(-)}=0<\lambda_1^{(+)}$, and all eigenvalues are simple. Finally, for $n=0$ the odd eigenfunction is $\operatorname{sgn}(z)\,y^\gamma(1-y)^{1/(r+1)}$, which by \eqref{eq:firstint} is proportional to $W'$, in agreement with $\Ss W'=0$; this identifies the kernel and provides an internal consistency check of the reduction.
\end{proof}

\begin{corollary}[The case $p=2$]
\label{cor:p2}
For $p=2$ one has $\lambda_0^{(+)}=-(2r+5)/(8(r+1))$, $\lambda_1^{(+)}=3r(r+2)/(8(r+1))$, $\lambda_n^{(-)}=(r+3)n/4+(r+1)n^2/8$, and the first four eigenvalues are $-7/16$, $0$, $9/16$, $5/4$ for $r=1$; $-11/32$, $0$, $45/32$, $2$ for $r=3$; and $-5/16$, $0$, $35/16$, $11/4$ for $r=5$.
\end{corollary}

The negative eigenvalue decreases in absolute value as $r$ grows while the positive gap increases, so nothing in $\Ss$ distinguishes large $r$ as more unstable. We stress again that $\Ss$ is not the generator of the linearised flow. The generator $a\partial_z\Ss$ is not self-adjoint in any obvious weighted space, and its spectral analysis is outside the scope of this paper; see Section~\ref{sec:disc}.

\section{Numerical verification and dynamical diagnostics}
\label{sec:num}

All computations in this section are reproducible with the scripts and data deposited at \url{https://github.com/FrancisRVillatoro/gradient-power-compactons}.  The exact release used for this manuscript is version 1.0.0, archived on Zenodo under DOI \href{https://doi.org/10.5281/zenodo.22671976}{10.5281/zenodo.22671976}.  The calculations use double precision and standard scientific-library routines (SciPy) for beta functions and sparse eigenvalue problems.

\subsection{Profiles, widths and edge exponents}

The profile \eqref{eq:profile} was evaluated for $p=2$, $c=\mu=1$ and $r=1,3,5$ by inverting the regularised incomplete beta function on a fine monotone grid. The half-widths obtained from \eqref{eq:width} are $\ell=2\pi$, $\ell=3.277905$ and $\ell=2.685137$, and a direct numerical quadrature of \eqref{eq:quad} reproduces them to six digits. A least-squares fit of $\log U$ against $\log d$ for $d\in[10^{-6},10^{-3}]$ gives slopes $1.9942$, $0.7998$ and $0.6666$, against the exact exponents $2$, $4/5$ and $2/3$ of \eqref{eq:edge}; the small deviation for $r=1$ is the expected effect of the $o(1)$ correction in \eqref{eq:edge} on the wider fitting window, and it decreases as the window is moved towards the boundary. Figure~\ref{fig:profiles} shows the profiles and the edge behaviour.

\subsection{Spectrum of the Sturm--Liouville factor}

The operator $\Ss$ of \eqref{eq:gen} was discretised for $p=2$, $r=3$, $c=\mu=1$ on a uniform grid of $N$ points in $[-\ell,\ell]$ by a conservative second-order finite-difference scheme, with the coefficient $b$ evaluated at the midpoints, homogeneous Dirichlet conditions at the endpoints and no special treatment of the crest. The eight smallest eigenvalues were computed by shift-and-invert Lanczos iteration and classified by parity. Table~\ref{tab:spectrum} compares them with \eqref{eq:even} and \eqref{eq:odd}.

\begin{table}[t]
\centering
\caption{Eigenvalues of $\Ss$ for $p=2$, $r=3$, $c=\mu=1$ computed by finite differences with $N$ grid points, against the exact values of Theorem~\ref{thm:spectrum}. Even eigenvalues are marked with $+$ and odd ones with $-$.}
\label{tab:spectrum}
\small\begin{tabular}{@{}lrrrrrrrr@{}}
\toprule
 & $\lambda_0^{(+)}$ & $\lambda_0^{(-)}$ & $\lambda_1^{(+)}$ & $\lambda_1^{(-)}$ & $\lambda_2^{(+)}$ & $\lambda_2^{(-)}$ & $\lambda_3^{(+)}$ & $\lambda_3^{(-)}$\\
\midrule
exact & $-0.34375$ & $0$ & $1.40625$ & $2$ & $4.15625$ & $5$ & $7.90625$ & $9$\\
$N=6001$ & $-0.34375$ & $0.01672$ & $1.40628$ & $2.03843$ & $4.15642$ & $5.06525$ & $7.90682$ & $9.09674$\\
$N=12001$ & $-0.34375$ & $0.01364$ & $1.40626$ & $2.03128$ & $4.15632$ & $5.05296$ & $7.90647$ & $9.07825$\\
$N=24001$ & $-0.34375$ & $0.01146$ & $1.40625$ & $2.02623$ & $4.15628$ & $5.04435$ & $7.90633$ & $9.06540$\\
$N=48001$ & $-0.34375$ & $0.00997$ & $1.40625$ & $2.02280$ & $4.15626$ & $5.03851$ & $7.90628$ & $9.05673$\\
\bottomrule
\end{tabular}
\end{table}

The even eigenvalues converge rapidly towards the exact values, including the negative eigenvalue $-11/32$. The odd eigenvalues converge from above at a substantially slower rate, and the available refinements do not establish a definite asymptotic order for either parity;  the numerical approximation of the zero eigenvalue is a small positive number that decreases with $N$. This behaviour is explained by \eqref{eq:oddfun}. The odd eigenfunctions contain the factor $(1-y)^{1/(r+1)}\sim|z|^{1/r}$ at the crest, so their derivative is unbounded there, and a uniform grid with no adaptation resolves the corresponding flux poorly; the even eigenfunctions are smooth at the crest. A discretisation in the variable $y$ with Jacobi weights, or a graded mesh at the crest, removes this effect, and we regard the slow drift of the odd eigenvalues towards their exact values as an independent confirmation of the crest regularity predicted by Theorem~\ref{thm:tw}(v) rather than as a numerical anomaly. In particular, the small positive value obtained for $\lambda_0^{(-)}$ must not be read as a positive eigenvalue of $\Ss$.

\subsection{Time-dependent propagation, Fourier regularity and the generator}
\label{sec:time}

We finally test whether the exact cusped profiles can be propagated by a direct spectral discretisation of \eqref{eq:E}, and whether the loss of regularity predicted in Section~\ref{sec:tw} is visible in Fourier space.  We set $p=2$, $c=\mu=1$ and consider $r=3$ and $r=5$ on the periodic interval $[-30,30)$.  To avoid confusing aliasing with the intrinsic low regularity of the compacton, the time-dependent calculation is formulated as a padded Fourier--Galerkin method.  On an $N$-point grid we retain the modes $|k|\le K=\lfloor N/3\rfloor$ and denote the corresponding projection by $P_K$.  The state is initialised as $u_N(0)=P_KU$.  Each nonlinear flux is evaluated on an oversampled grid of $M=3N$ points and projected back to the retained space:
\begin{equation}
 \dot u_N=-D_N P_KF_2(u_N)-D_N^2P_K\Psi_r\!\left(D_NP_KF_2(u_N)\right).
 \label{eq:semi-discrete}
\end{equation}
Here the padded evaluation is applied both to $F_2(u)=|u|u$ and to $\Psi_r(q)=|q|^{r-1}q$.  For the polynomial inner powers $r=3,5$, the factor $M/N=3$ is sufficient to prevent unresolved products from aliasing into the retained band; because $F_2$ is non-polynomial once small Gibbs undershoots occur, we also checked convergence with respect to the padding factor.  At $N=384$ and $T=0.1$, changing $M/N$ through $2,2.5,3,4,6$ changes the dynamic relative $L^2$ error by only $6.1\times10^{-6}$ for $r=3$ and $1.6\times10^{-6}$ for $r=5$.  No hyperviscosity or post-step filtering is used.

Time integration is classical fourth-order Runge--Kutta with
\begin{equation}
 \Delta t=5\times10^{-5}\left(\frac{256}{N}\right)^3,
 \label{eq:dt}
\end{equation}
consistent with the third-order dispersive scaling.  Since a nonsmooth exact profile cannot be represented spectrally without error, we separate representation and evolution errors.  At each output time we determine the scalar phase correction $\delta$ that minimises
\[
 \|u_N-P_KU(\,\cdot-t-\delta)\|_2,
\]
and define
\begin{equation}
 E_{\rm dyn}=\frac{\|u_N-P_KU(\cdot-t-\delta)\|_2}{\|P_KU\|_2},\qquad
 E_{\rm tot}=\frac{\|u_N-U(\cdot-t-\delta)\|_2}{\|U\|_2}.
 \label{eq:dyn-errors}
\end{equation}
The corresponding representation error is obtained by replacing $u_N$ with $P_KU$ in the second expression.  This decomposition prevents the algebraic Fourier tail of the exact cusp from being counted as a time-integration error.  We also monitor the invariants $M$ and $Q$ of Section~\ref{sec:family} and the phase correction itself.

A temporal-refinement test at $N=384$ and $T=0.1$ shows that the error is spatially dominated.  Replacing \eqref{eq:dt} successively by $\Delta t/2$ and $\Delta t/4$ changes $E_{\rm dyn}$ by only $1.1\times10^{-9}$ for $r=3$ and $2.4\times10^{-9}$ for $r=5$.  With the baseline time step one obtains
\[
 E_{\rm dyn}=9.886\times10^{-3}\quad(r=3),\qquad
 E_{\rm dyn}=7.142\times10^{-3}\quad(r=5),
\]
while the relative drifts of $Q$ are $-1.87\times10^{-7}$ and $2.74\times10^{-8}$, respectively, and the mass drift is at roundoff level.  Table~\ref{tab:propagation} gives a spatial refinement at the same final time.  The errors are not monotone enough to assign an empirical convergence order, which is unsurprising for a uniform Fourier grid applied to a compactly supported profile with algebraic endpoint singularities, but both the dynamic and total errors decrease overall.

\begin{table}[t]
\centering
\caption{Phase-aligned propagation errors at $T=0.1$ for the padded Fourier--Galerkin scheme \eqref{eq:semi-discrete}.  $E_{\rm dyn}$ compares within the retained spectral space, whereas $E_{\rm tot}$ also contains the representation error of the exact compacton.}
\label{tab:propagation}
\begin{tabular}{@{}rcccc@{}}
\toprule
& \multicolumn{2}{c}{$r=3$} & \multicolumn{2}{c}{$r=5$}\\
$N$ & $E_{\rm dyn}$ & $E_{\rm tot}$ & $E_{\rm dyn}$ & $E_{\rm tot}$\\
\midrule
128 & $3.114\times10^{-2}$ & $3.646\times10^{-2}$ & $2.898\times10^{-2}$ & $4.145\times10^{-2}$\\
256 & $7.500\times10^{-3}$ & $8.516\times10^{-3}$ & $1.539\times10^{-2}$ & $2.192\times10^{-2}$\\
384 & $9.886\times10^{-3}$ & $1.023\times10^{-2}$ & $7.142\times10^{-3}$ & $1.112\times10^{-2}$\\
512 & $5.671\times10^{-3}$ & $6.602\times10^{-3}$ & $8.244\times10^{-3}$ & $1.007\times10^{-2}$\\
\bottomrule
\end{tabular}
\end{table}

Longer integrations at $N=384$ remain coherent up to $T=1$.  The maximum dynamic errors recorded at ten equally spaced times are $1.10\times10^{-2}$ for $r=3$ and $1.46\times10^{-2}$ for $r=5$, with final values $7.20\times10^{-3}$ and $1.02\times10^{-2}$; the corresponding maximum total errors are $1.21\times10^{-2}$ and $1.61\times10^{-2}$.  The maximum absolute phase corrections are $3.20\times10^{-3}$ and $9.16\times10^{-3}$, with final corrections $-3.20\times10^{-3}$ and $-4.49\times10^{-3}$.  A least-squares fit of $\delta(t)$ over the sampled interval corresponds to effective speed biases $-1.17\times10^{-3}$ and $-5.99\times10^{-3}$ relative to $c=1$, respectively; the phase oscillates, so these slopes are used only as diagnostics.  The maximum absolute drift of $Q$ is $2.97\times10^{-7}$ for $r=3$ and $1.52\times10^{-6}$ for $r=5$, while the maximum mass drifts are $2.7\times10^{-13}$ and $1.9\times10^{-12}$.  Small negative Gibbs undershoots remain present near the free boundaries (approximately $-2.7\times10^{-2}$ and $-5.1\times10^{-2}$ at this resolution), so the calculation is a propagation and consistency test rather than evidence for positivity preservation or nonlinear stability.  Figures~\ref{fig:time-error} and~\ref{fig:phase-error} show the shape and phase diagnostics.

\begin{figure}[t]
\centering
\includegraphics[width=.72\linewidth]{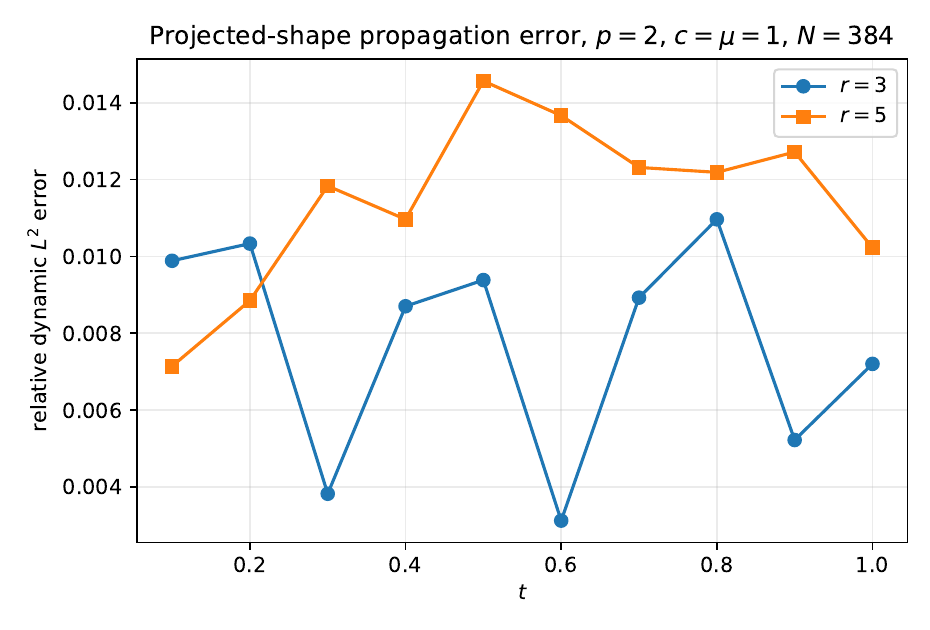}
\caption{Relative dynamic $L^2$ error $E_{\rm dyn}$ for $p=2$, $c=\mu=1$, $N=384$ and the padded Fourier--Galerkin discretisation.  The compactons remain coherent for $0\le t\le1$ at percent-level error.}
\label{fig:time-error}
\end{figure}

\begin{figure}[t]
\centering
\includegraphics[width=.72\linewidth]{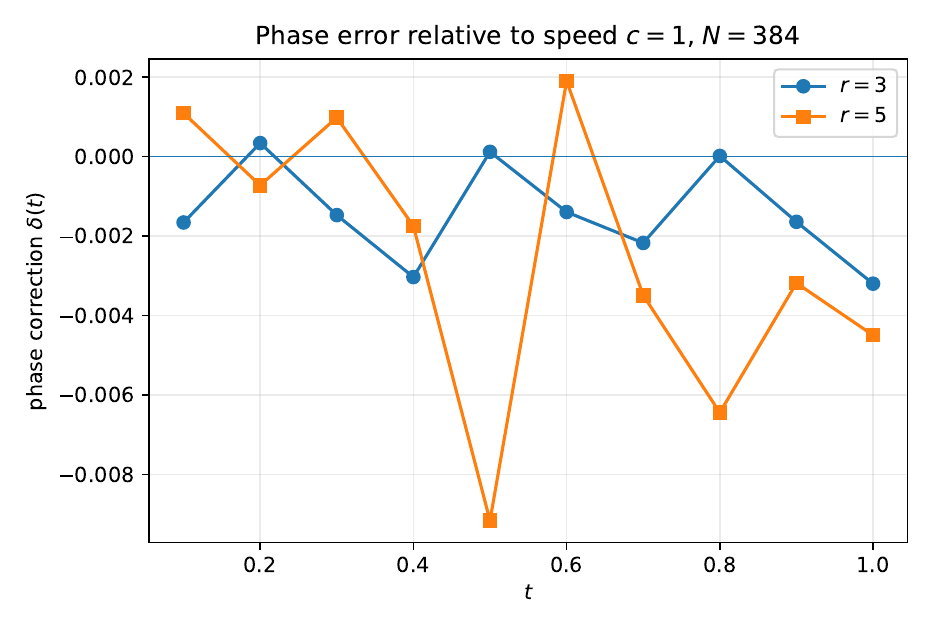}
\caption{Phase correction $\delta(t)$ used in \eqref{eq:dyn-errors}.  Its magnitude remains below $3.2\times10^{-3}$ for $r=3$ and $9.2\times10^{-3}$ for $r=5$ over the unit-time calculation.}
\label{fig:phase-error}
\end{figure}

The absence of spectral convergence can be quantified independently of the time integrator.  In fact, the edge expansion in Theorem~\ref{thm:tw}(iv) already predicts the leading Fourier asymptotic.  For an even compacton with
\[
 U(x)=C_e(\ell-|x|)^\alpha+o((\ell-|x|)^\alpha)
\]
at both edges, standard endpoint asymptotics give, for large positive wavenumber $k$,
\begin{equation}
 \widehat U(k)=2C_e\Gamma(\alpha+1)k^{-\alpha-1}
 \cos\!\left(k\ell-\frac{\pi(\alpha+1)}{2}\right)+o(k^{-\alpha-1}),
 \label{eq:fourier-asymptotic}
\end{equation}
up to the Fourier-normalisation convention.  Thus the coefficient envelope is $O(k^{-1-\alpha})$ and the $L^2$ projection tail is $O(N^{-\alpha-1/2})$.  The crest singularity, $A-U=O(|z|^{1+1/r})$, instead contributes the faster decay $O(k^{-2-1/r})$, so the free boundary controls the Fourier tail for the cases considered here.

Sampling the exact profiles on a grid of $2^{18}$ points and fitting the Fourier coefficients over $500\le k\le5000$ gives
\begin{equation}
 |\widehat U_k|\sim k^{-1.80002}\quad (r=3),
 \qquad
 |\widehat U_k|\sim k^{-1.66641}\quad (r=5).
 \label{eq:fourier-fits}
\end{equation}
These agree with $-1-\alpha=-9/5$ and $-5/3$.  The fitted slopes are stable under changes of the fitting window: using $250$--$2500$, $500$--$5000$, $1000$--$10000$ and $2000$--$20000$ gives values between $-1.8078$ and $-1.7999$ for $r=3$, and between $-1.6725$ and $-1.6642$ for $r=5$.  The corresponding relative $L^2$ Fourier-projection errors, computed directly from the tail of the finely sampled spectrum, scale as
\begin{equation}
 E_N^{\rm proj}\sim N^{-1.2978}\quad(r=3),
 \qquad
 E_N^{\rm proj}\sim N^{-1.1648}\quad(r=5),
 \label{eq:projection-rates}
\end{equation}
again in agreement with $N^{-13/10}$ and $N^{-7/6}$.  Figure~\ref{fig:fourier-decay} shows the coefficient decay.  Thus the slower Fourier convergence for $r=5$ is an intrinsic consequence of the exact edge regularity rather than a failure of the time stepper.

\begin{figure}[t]
\centering
\begin{minipage}[t]{0.485\linewidth}
  \centering
  \includegraphics[width=\linewidth]{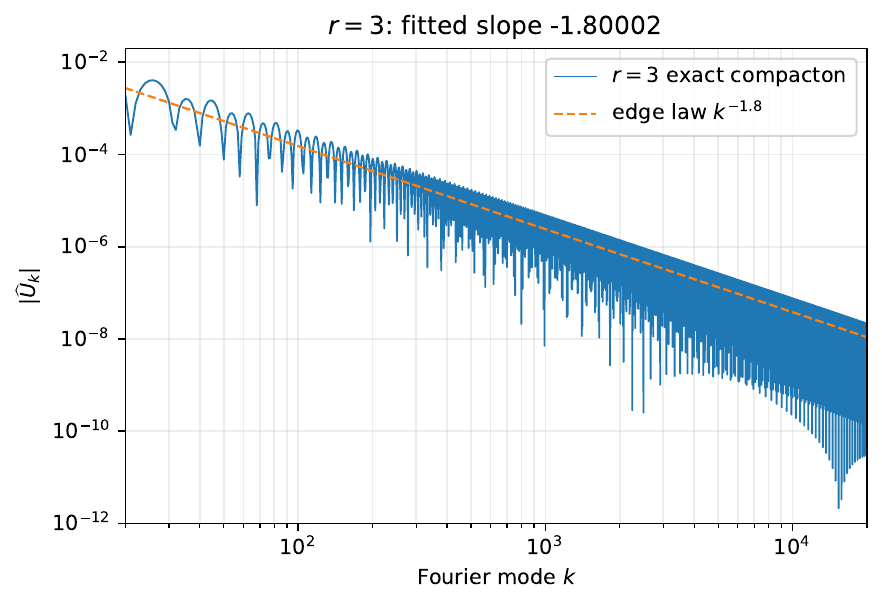}\\[-1mm]
  {\small (a) $r=3$}
\end{minipage}\hfill
\begin{minipage}[t]{0.485\linewidth}
  \centering
  \includegraphics[width=\linewidth]{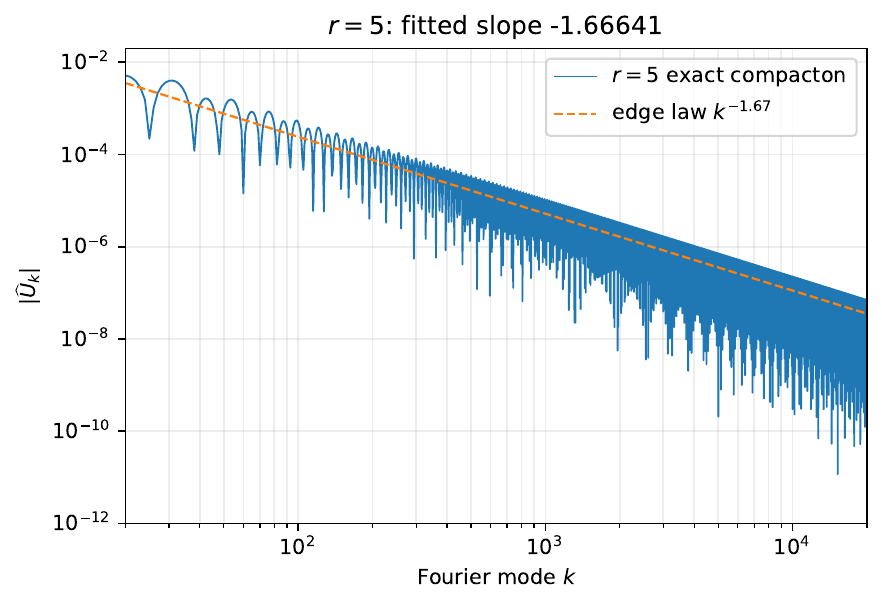}\\[-1mm]
  {\small (b) $r=5$}
\end{minipage}
\caption{Fourier coefficients of the exact $p=2$ compactons shown separately to make the algebraic tails visible.  The fitted slopes are $-1.80002$ for $r=3$ and $-1.66641$ for $r=5$, in agreement with the endpoint asymptotic \eqref{eq:fourier-asymptotic}.  The free-boundary singularity therefore controls the high-frequency tail in both cases.}
\label{fig:fourier-decay}
\end{figure}

As a separate linear diagnostic we project the generator
\begin{equation}
 \mathcal A=a\,\partial_z\Ss
 \label{eq:generator-num}
\end{equation}
onto the exact Jacobi eigenfunctions of $\Ss$ derived in Theorem~\ref{thm:spectrum}, retaining the same number of even and odd modes and evaluating the coupling integrals by high-order Gaussian quadrature.  This basis is deliberately structure adapted, so the calculation is used only as a diagnostic.  The resulting finite matrices are not self-adjoint.  Nevertheless, for matrix dimensions between $12$ and $48$ every computed non-zero eigenvalue occurs numerically in a pair $\pm i\omega$, and the largest real part remains at floating-point roundoff.  At dimension $48$ we obtain
\begin{equation}
 \max_{\lambda\in\sigma(\mathcal A_{48})}|\Re\lambda|
 =1.32\times10^{-13}\quad(r=3),
 \qquad
 =4.55\times10^{-13}\quad(r=5).
 \label{eq:generator-realparts}
\end{equation}
For example, the first three positive imaginary parts are $(2.236,10.214,25.876)$ for $r=3$ and $(4.393,20.602,53.261)$ for $r=5$ at the finest truncation.  Figure~\ref{fig:generator-realparts} displays the real-part diagnostic under refinement.  This concerns a truncated representation of the point spectrum only: it is evidence for modal neutrality within this structure-adapted Galerkin calculation, not a spectral-stability theorem, and it gives no information about possible high-frequency transient amplification associated with the degenerate coefficients.

\begin{figure}[t]
\centering
\includegraphics[width=.68\linewidth]{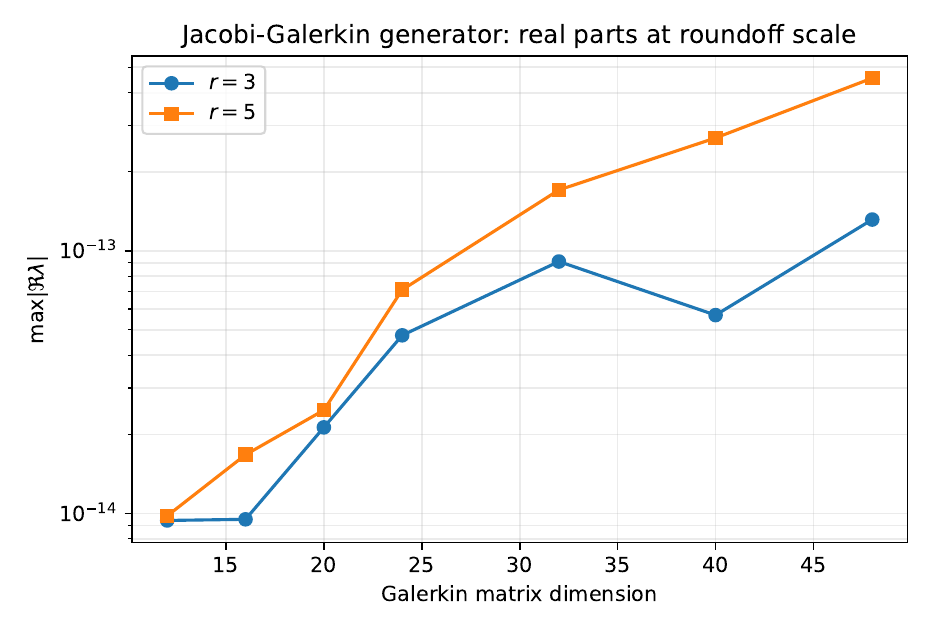}
\caption{Maximum absolute real part of the eigenvalues of the structure-adapted Jacobi--Galerkin truncation of $a\partial_z\Ss$.  Values remain at the level of floating-point roundoff as the truncation is increased.}
\label{fig:generator-realparts}
\end{figure}

\section{Discussion}
\label{sec:disc}

The main outcome of this work is a separation of mechanisms that coincide in the classical Rosenau--Hyman compacton.  The gradient-power exponent $r$ leaves the amplitude--speed relation unchanged but destroys the amplitude-independent width as soon as $r\neq1$.  At the same time, two different transition curves govern two different geometric questions: $pr=1$ determines whether the solitary wave reaches the vacuum at finite distance, whereas $r(p-1)=2$ determines whether that free boundary is regular or edge-cusped.  Compact support and divergent edge slope are therefore not manifestations of a single threshold.  The family \eqref{eq:E} makes this distinction explicit because both transitions follow from the same closed travelling-wave quadrature.

The weak-solution result shows that the geometric singularity is milder for the flux than for the profile.  Even when $|U_z|\to\infty$ at the free boundary, the quantities that enter the conservation law vanish with sufficient regularity to prevent singular measures; similarly, the loss of $C^2$ regularity at the crest for $r>1$ does not create a distributional defect.  This separation between profile regularity and flux regularity is also visible numerically.  The high-frequency Fourier tail is controlled by the exact edge exponent, not by the crest singularity: the measured laws $k^{-9/5}$ and $k^{-5/3}$ for $r=3$ and $r=5$ agree with $k^{-1-\alpha}$, and the corresponding projection errors agree with $N^{-\alpha-1/2}$.  Thus the free-boundary geometry quantitatively predicts the loss of spectral convergence.  This provides a direct bridge between an exact singularity of the nonlinear wave and the approximation properties of a global numerical method.

The singular linearisation exhibits a complementary separation.  For $r>1$ the coefficient $b$ vanishes at the crest like $|z|^{(r-1)/r}$, while the compacton endpoints are singular Sturm--Liouville endpoints.  Nevertheless, after the compacton-adapted coordinate change the natural Friedrichs realisation is exactly of Jacobi type.  The crest degeneracy is encoded by the transmission condition and changes the spectrum continuously with $r$, but it creates no additional negative direction: $\Ss$ has exactly one negative eigenvalue and a simple translation kernel for the whole range covered by Theorem~\ref{thm:spectrum}.  The independence of these eigenvalues from $c$ and $\mu$ further shows that the Morse index is controlled by the exponents rather than by the scale of the wave.

This self-adjoint information should not be confused with stability of the compacton. Classical compacton stability criteria rely on additional variational or Hamiltonian structure \cite{DeyKhare1998}, whereas the evolution generator here is $a\partial_z\Ss$, not $\Ss$ itself, and no Hamiltonian index theorem is available for \eqref{eq:E}.  The Jacobi--Galerkin calculations show no detectable real part in the truncated non-zero point spectrum, but this is only a modal diagnostic.  Degenerate KdV-type equations can exhibit strong high-frequency amplification or ill-posedness in standard Sobolev spaces even when a conventional point-spectrum calculation is inconclusive \cite{AmbroseSimpsonWrightYang2012,JeongOh2026}; related models can recover positive well-posedness statements after solution-adapted changes of variables \cite{GermainHarropGriffithsMarzuola2019}.  Equation~\eqref{eq:E} is particularly interesting in this respect because its generator degenerates at two geometrically distinct locations: at the free boundary for all compactons and, when $r>1$, at the crest where $\partial_z\Fp(U)=0$.  Determining whether these two degeneracies produce different high-frequency mechanisms is a natural continuation of the present work.

Several further directions remain open.  A standard Hamiltonian formulation is not evident for $r\neq1$; the generalized Lagrange formulation known for the absolute $|K|(p,p)$ equation \cite{GarralonRusVillatoro2023CSF} does not provide the Hamiltonian variational structure of the Cooper--Shepard--Sodano family.  Compacton interactions are another natural test.  Since the width now varies with amplitude, collisions between waves of different speeds no longer compare equal-width objects, and the robust compacton--anticompacton interactions observed for $r=1$ \cite{GarralonRusVillatoro2023CSF,GarralonRusVillatoro2023CNSNS} need not persist in the edge-cusped regime.  Numerically, methods adapted to the exact edge singularity or to the Jacobi coordinate are likely to be more informative than global Fourier schemes once high accuracy is required.

Taken together, the value of \eqref{eq:E} is not only that it supplies another explicit compacton family.  It provides a controlled nonlinear-wave model in which compactification, edge regularity, weak admissibility, singular spectral structure and numerical convergence can be varied and compared within one exactly solvable framework.  The two critical curves and the explicit spectral factor make the family a useful benchmark for the broader analysis of degenerate dispersive waves with free boundaries.

\section*{Acknowledgements}
This work was supported by Universidad de M\'alaga, project PPRO-B4-2026-004 (Plan Propio de Investigaci\'on, Transferencia y Divulgaci\'on Cient\'{\i}fica). The author acknowledges the computational resources, technical expertise, and assistance provided by the SCBI (Supercomputing and Bioinformatics Center) at the Universidad de M\'alaga, including access to the Picasso Supercomputer.

During the preparation and revision of this manuscript, the author used OpenAI ChatGPT (GPT-5.6 Sol) and Anthropic Claude (Opus 5) for language revision, mathematical checking, code review and assistance with numerical post-processing. All AI-assisted output was critically reviewed and edited by the author, who takes full responsibility for the content of the manuscript.

\section*{Data availability}
All scripts, figures and numerical data needed to reproduce the computational results are openly available in the development repository on \href{https://github.com/FrancisRVillatoro/gradient-power-compactons}{GitHub}.  The exact reproducibility release used for this manuscript is version 1.0.0, archived on Zenodo (DOI \href{https://doi.org/10.5281/zenodo.22671976}{10.5281/zenodo.22671976}) \cite{VillatoroGradientPowerCode2026}.

\bibliographystyle{iopart-num}
\bibliography{refs}

\end{document}